\documentclass[a4paper]{article}
\usepackage{amsmath}

\newtheorem{theorem}{Theorem}

\newtheorem{definition}{Definition}

\newtheorem{remark}[theorem]{Remark}
\newtheorem{example}{Example}

\newenvironment{proof}[1][Proof]{\noindent\textbf{#1.} }{\ \rule{0.5em}{0.5em}}

\begin{document}

\title{\textbf{Perfect Mannheim, Lipschitz and Hurwitz weight codes }}
\author{Murat G\"{u}zeltepe  \\
{\small Department of Mathematics, Sakarya University, TR54187 Sakarya,
Turkey}}
\date{}
\maketitle

\begin{abstract}
In this paper, upper bounds on codes over Gaussian integers,
Lipschitz integers and Hurwitz integers with respect to Mannheim
metric, Lipschitz and Hurwitz metric are given.
\end{abstract}




\section{Introduction }

\ If a code attains an upper bound (the sphere-packing bound) in a
given metric, then it is called a perfect code. Perfect codes have
always drawn the attention of coding theorists and mathematicians
since they play an important role in coding theory for theoretical
and practical reasons. All perfect codes with respect to Hamming
metric over finite fields are known \cite{1}-\cite{4}. For
non-field alphabets only trivial codes are known and by similar
methods it was proved in \cite{5}.

\ Perfect codes have been investigated not only with respect to
Hamming metric but also other metrics, for example Lee metric. Lee
metric was introduced in \cite{Lee}. Some perfect codes with
respect to Lee metric were discovered in \cite{6}.

Later, Mannheim metric was introduced by Huber in  \cite{7}. It is
well known that the Euclidean metric is the relevant metric for
maximum-likelihood decoding. Although Mannheim metric is a
reasonable approximation to it, it is not a priori, a natural
choice. However, the codes being proposed are very useful in coded
modulation schemes based on quadrature amplitude modulation
(QAM)-type constellations for which neither Hamming nor Lee metric
is appropriate. Two classes of codes over Gaussian integers $G$
were considered in \cite{7}, namely, the one Mannheim
error-correcting codes (OMEC), and codes having minimum Mannheim
distance greater than three. The OMEC codes are perfect with
respect to Mannheim metric. Thus, some perfect codes were
discovered. But, dimension $k$ of OMEC codes with parameters
$[n,k,d]$ are only $n-1$. In the present study, we obtain some
perfect codes with respect to Mannheim metric. The dimension of
these perfect codes are not only $n-1$ but also $n-k, \ (k<n)$.

On the other hand, Lipschitz metric was presented and some perfect
codes over Lipschitz integers with respect to Lipschitz metric
were introduced in \cite{8,9}.

 In this paper, we consider the existence and nonexistence of perfect codes with respect to
 Mannheim metric and Lipschitz metric over Gaussian integers, Lipschitz integers and Hurwitz integers.
 Also, we introduce Hurwitz metric and we give upper bounds on these codes over Hurwitz integers.\\\\

In what follows, we consider the following:

\begin{definition} \cite{6} An $(n, k)$ linear code is said to be perfect if for a given positive integer
$t$, the code corrects all errors of weight $t$ or less and no
error of weight greater than $t$. For a perfect code correcting
errors of weight $t$ or less, number of vectors of weight $t$ or
less including the vector of all zeros is equal to the number of
available cosets.

\end{definition}

\begin{definition} \cite{7,9}
Let $G$ be denotes the set of all Gaussian integers and $G_\pi$,
the residue class of $G$ modulo $\pi $, where $\pi {\pi ^
* } = p \equiv 1\quad (\bmod \;4)$ and $\pi^*$ is conjugate of
$\pi$. For $\beta $, $\gamma
 \in G_\pi $, consider $a+bi$ in the class of $\beta - \gamma $ with $|a|+|b|$ minimum. Mannheim distance $d_M$ between $\beta$
and $\gamma $ is $$d_M(\beta ,\gamma)=|a|+|b| .$$
\end{definition}

Note that Mannheim distance is not a true metric. The metric given
by Def. (2) is a true metric \cite{9}. We will use this metric as
Mannheim metric in the present paper.

 More information which are related with Mannheim metric
and Mannheim weight can be found in \cite{7,8,9}.

\begin{definition} \cite{10} The Hamilton Quaternion Algebra over the set of the real numbers
($\mathcal{R} $), denoted by $H(\mathcal{R})$, is the associative
unital algebra given by the following representation:

i)$H(\mathcal{R})$ is the free $\mathcal{R}$ module over the
symbols $1,e_1,e_2,e_3$, that is, $ H(\mathcal{R}) = \{ a_0  + a_1
e_1 +$ $ a_2 e_2 + a_3 e_3:\;a_0 ,a_1 ,a_2 ,a_3  \in
\mathcal{R}\}$;

ii)1 is the multiplicative unit;

iii) $ e_1^2  = e_2^2  = e_3^2  =  - 1$;

iv) $ e_1e_2 =  - e_2e_1 = e_3,\;e_3e_1 =  - e_1e_3 = e_2,\;e_2e_3
= - e_3e_2 = e_1$ .
\end{definition}

The set of Lipschitz integers $H(\mathcal{Z})$, which is defined
by $H(\mathcal{Z})) = \left\{ {{a_0} + {a_1}{e_1} + } \right.$
$\left. {{a_2}{e_2} + {a_3}{e_3}:{a_0},{a_1},{a_2},{a_3} \in
\mathcal{Z}} \right\}$, is a subset of $H(\mathcal{R})$, where
$\mathcal{Z}$ is the set of all integers. If $ q = a_0  + a_1 e_1
+ a_2 e_2 + a_3 e_3$ is a quaternion integer, its conjugate
quaternion is $ q^* = a_0 - (a_1 e_1 + a_2 e_2 + a_3 e_3)$. The
norm of  $q$ is $ N(q) = q q^* = a_0^2  + a_1^2 + a_2^2 + a_3^2$.
The units of ${H(\mathcal{Z})}$ are $ \pm 1,\pm e_1, \pm e_2,\pm
e_3$.

\begin{definition} \cite{10} Let $ \pi $ be an odd. If
there exist $\delta \in{H(\mathcal{Z})}$ such that $q_1-q_2=\delta
\pi$ then $q_1,q_2 \in {H(\mathcal{Z})}$ are right congruent
modulo $\pi$ and it is denoted as $q_1 \equiv_rq_2$.

\end{definition}

This equivalence relation is well-defined. We can consider the
ring of the quaternion integers modulo this equivalence relation,
which we denote as
$$H(\mathcal{Z})_\pi=\left\{ {\left. {q\;(\bmod \pi )} \right|\;q
\in H(\mathcal{Z})} \right\} \ \cite{9}.$$

Except as noted otherwise, we will use right congruent modulo
$\pi$ in the present paper. Analogous result hold for left
congruent modulo $\pi$.

\begin{theorem} \cite{8} Let $\alpha \in H(\mathcal{Z})$. Then $H(\mathcal{Z})_\alpha$ has
$(N(\alpha))^2$ elements .
\end{theorem}

\begin{definition}\cite{9} Let $ \pi $ be a quaternion integer. Given $\alpha, \beta \in
H(\mathcal{Z})_\pi$, then Lipschitz distance between $\alpha$ and
$\beta$ is computed as $\left| {{a_0}} \right| + \left| {{a_1}}
\right| + \left| {{a_2}} \right| + \left| {{a_3}}
 \right|$ and denoted by $d_L(\alpha,\beta)$, where $$\alpha  - \beta { \equiv _r}{a_0} + {a_1}e_1 + {a_2}e_2 +
 {a_3}e_3 \ (mod\ \pi)$$
 with $\left| {{a_0}} \right| + \left| {{a_1}}
\right| + \left| {{a_2}} \right| + \left| {{a_3}}
 \right|$ minimum.
\end{definition}

Lipschitz weight of the element $\gamma$ is defined as  $\left|
{{a_0}} \right| + \left| {{a_1}} \right| + \left| {{a_2}} \right|
+ \left| {{a_3}}
 \right|$ and is denoted by $w_{L}(\gamma)$, where
 $\gamma=\alpha-\beta$ with $\left| {{a_0}} \right| + \left| {{a_1}}
\right| + \left| {{a_2}} \right| + \left| {{a_3}}
 \right|$ minimum.

More information which are related with the arithmetic properties
of $H(\mathcal{Z})$ can be found in \cite{9,10}.

\begin{definition} \cite{11} The set of all Hurwitz integers is $$\mathcal{H}=\left\{ {{a_0} + {a_1}{e_1} + {a_2}{e_2} +
{a_3}{e_3}\in H(\mathcal{R}):{a_0},{a_1},{a_2},{a_3} \in
\mathcal{Z}\ {\rm{ or }} \ {a_0},{a_1},{a_2},{a_3} \in \mathcal{Z}
+ \frac{1}{2}} \right\}.$$

\end{definition}
It can be checked that $\mathcal{H}$ is closed under quaternion
multiplication and addition, so that it forms a subring of the
ring of all quaternions. The units of $\mathcal{H}$ are $\pm1,\pm
e_1,\pm e_2, \pm e_3, \pm \frac{1}{2}\pm \frac{1}{2}e_1\pm
\frac{1}{2}e_2\pm \frac{1}{2}e_3$.

\begin{definition} Let $ \pi$ be a prime in $H(\mathcal{Z})$. If
there exists $\delta \in H(\mathcal{Z})$ such that $q_1-q_2=\delta
\pi$ then $q_1,q_2 \in \mathcal{H}$ are right congruent modulo
$\pi$ and it is denoted as $q_1 \equiv_rq_2$.

\end{definition}

This equivalence relation is well-defined. We can consider the
ring of the Hurwitz integers modulo this equivalence relation,
which we denote as
$$\mathcal{H}_\pi=\left\{ {\left. {q\;(\bmod \pi )} \right|\;q
\in \mathcal{H}} \right\}.$$

\begin{theorem} Let $\alpha$ be a prime integer quaternion. Then $\mathcal{H}_\alpha$ has
$2N(\alpha)^2-1$ elements.
\end{theorem}

\begin{proof} Let $\pi0$ be a prime integer quaternion. According to Theorem 1, the cardinal number of
$ H(\mathcal{Z})_\pi$ is equal to $N(\pi)^2$. Also, the cardinal
number of $ H(\mathcal{Z}+\frac{1}{2})_\pi$ is equal to
$N(\pi)^2$. $( H(\mathcal{Z})_\pi -\left\{ {0} \right\}) \cap
(H(\mathcal{Z}+\frac{1}{2})_\pi-\left\{ {0} \right\})=\emptyset$
since the elements of the set  $
H(\mathcal{Z}+\frac{1}{2})_\pi-\left\{ {0} \right\}$ are defined
in the form $q-\delta \pi = a_0+a_1e_1+a_2e_2+a_3e_3+a_4w$, where
$q\in H(\mathcal{Z}+\frac{1}{2}),$ $\delta,\;\pi \in
H(\mathcal{Z})$, $\ a_0,a_1,a_2,a_3\in \mathcal{Z}$ and $a_4$ is
an odd integer. But the additive identity is an element of both
sets $ H(\mathcal{Z})_\pi$ and $ H(\mathcal{Z}+\frac{1}{2})_\pi$.
Hence the proof is completed.

\end{proof}

Note that if $\delta$ is chosen from $\mathcal{H}$ instead of
$H(\mathcal{Z})$ then, Theorem 2 does not hold.

\ In the following definition, we introduce Hurwitz metric.

\begin{definition}Let $ \pi$ be a prime quaternion integer. Given $\alpha={a_0} + {a_1}{e}_1 + {a_2}{e}_2 +
 {a_3}{e}_3+a_4w , \beta={b_0} + {b_1}{e}_1 + {b_2}{e}_2 +
 {b_3}{e}_3+b_4w \in
\mathcal{H}_\pi$, then the distance between $\alpha$ and $\beta$
is computed as $\left| {{c_0}} \right| + \left| {{c_1}} \right| +
\left| {{c_2}} \right| + \left| {{c_3}}\right|+ \left| {{c_4}}
\right|$ and denoted by $d_H(\alpha,\beta)$, where $$\gamma
=\alpha - \beta { \equiv _r}{c_0} + {c_1}{e}_1 + {c_2}{e}_2 +
 {c_3}{e}_3+c_4w \ (mod\ \pi)$$
 with $\left| {{c_0}} \right| + \left| {{c_1}}
\right| + \left| {{c_2}} \right| + \left| {{c_3}}
 \right|+\left| {{c_4}} \right|$ minimum.
\end{definition} Also, we define Hurwitz weight of $\gamma =
\alpha -\beta$ as $$w_H(\gamma)=d_H(\alpha,\beta).$$
 It is possible to show that $d_H(\alpha, \beta)$ is a
 metric. We only show that the triangle inequality holds since the other conditions are straightforward. For this, let $\alpha$,
 $\beta$, and $\gamma$ be any three elements of $\mathcal{H}_\pi$. We
 have

 i) $d_H(\alpha, \beta)=w_H(\delta_1)=\left| {{a_0}} \right| + \left| {{a_1}}
\right| + \left| {{a_2}} \right| + \left| {{a_3}}
 \right|+ \left| {{a_4}} \right|  $, where $\delta_1 \equiv \alpha - \beta ={a_0} + {a_1}{e}_1 + {a_2}{e}_2 +
 {a_3}{e}_3+ {a_4}w \ (mod\  \pi )$ is an element of $\mathcal{H}_\pi$, and
 $\left| {{a_0}} \right| + \left| {{a_1}}
\right| + \left| {{a_2}} \right| + \left| {{a_3}}
 \right|+ \left| {{a_4}} \right|  $ is minimum.\

\

ii) $d_H(\alpha, \gamma)=w_H(\delta_2)=\left| {{b_0}} \right| +
\left| {{b_1}} \right| + \left| {{b_2}} \right| + \left| {{b_3}}
\right|+ \left| {{b_4}} \right| $, where $\delta_2 \equiv \alpha -
\gamma ={b_0} + {b_1}{e}_1 + {b_2}{e}_2 +
 {b_3}{e}_3 + {b_4}w \ (mod\  \pi )$ is an element of $\mathcal{H}_\pi$, and
 $\left| {{b_0}} \right| + \left| {{b_1}}
\right| + \left| {{b_2}} \right| + \left| {{b_3}}
 \right|+ \left| {{b_4}} \right|  $ is minimum.\

\

 iii) $d_H(\gamma, \beta)=w_H(\delta_3)=\left| {{c_0}} \right| +
\left| {{c_1}} \right| + \left| {{c_2}} \right| + \left| {{c_3}}
 \right| + \left| {{c_4}} \right|$, where $\delta_3 \equiv \gamma - \beta ={c_0} + {c_1}{e}_1 + {c_2}{e}_2 +
 {c_3}{e}_3+{c_4}w \ (mod\  \pi )$ is an element of $\mathcal{H}_\pi$, and
 $\left| {{c_0}} \right| +
\left| {{c_1}} \right| + \left| {{c_2}} \right| + \left| {{c_3}}
 \right| + \left| {{c_4}} \right|$ is minimum.\

 Thus, $\alpha-\beta=\delta_2+\delta_3 \ (\bmod \ \pi)$. However, $w_H\left( {{\delta _2} + {\delta _3}} \right) \ge w_H\left( {{\delta _1}}
 \right)$ since $w_H(\delta_1)=\left| {{a_0}} \right| + \left| {{a_1}}
\right| + \left| {{a_2}} \right| + \left| {{a_3}}
 \right|+ \left| {{a_4}}
 \right|$ is minimum. Therefore, $$d_H(\alpha, \beta) \le d_H(\alpha, \gamma)+d_H(\gamma, \beta).
 $$

Note that Hurwitz metric is not Lipschitz metric. To see this,
Lipschitz weight of the element
$w=\frac{1}{2}+\frac{1}{2}{e}_1+\frac{1}{2}{e}_2+\frac{1}{2}{e}_3$
is $w_L(w)=2$ and Hurwitz weight of the same element is
$w_H(w)=1$. \ \

The rest of this paper is organized as follows. In Section 2, an
upper bound on the number of parity check digits for linear
Mannheim weight codes correcting errors of Mannheim weight 1 and
Mannheim weight 2 or less over $G_\pi$($\pi \pi^*=p\ge5$ , a
prime) is obtained. Also, the bound with equality for the
existence of perfect codes is examined and an example of a perfect
code correcting errors of Mannheim weight 1 over $G_{2+i}$ is
given. In the third section of the present paper, a similar study
for linear Lipschitz weight codes correcting errors of Lipschitz
weight 1 and Lipschitz weight 2 or less over $H(\mathcal{Z})_\pi$
is given. In the fourth section, a similar study for linear
Lipschitz weight codes correcting errors of Lipschitz weight 1 and
Lipschitz weight 2 or less over $\mathcal{H}_\pi$ is presented. In fifth section, upper bounds on linear Hurwitz weight codes are defined. \\

\section{Perfect codes over Gaussian integers}

\subsection{Perfect codes correcting errors of Mannheim weight 1}

First, an upper bound on the number of parity check digits for one
Mannheim error correcting codes over $G_\pi$ ($p\equiv 1 \ (mod\
4)$) is optained. Note that a Mannheim error of weight 1 takes on
one of the four values $\pm 1,\ \pm i$, where $i^2=-1$.

\begin{theorem} An $(n, k)$ linear code over $G_\pi$ corrects all errors of Mannheim weight 1 provided that the bound $p^{n-k}\ge
4n+1$, where $p\equiv 1 \ (mod\ 4)$. Here and thereafter, $p$ will
denote an odd prime number.

\end{theorem}

\begin{proof} Error vectors of Mannheim weight 1 have just one nonzero
component. The nonzero component of the above stated error vectors
can take on one of the four values $\pm 1$, $\pm i$. The number of
errors of Mannheim weight 1 including the vector of all zeros over
$G_\pi$ is $4\left( {\begin{array}{*{20}{c}}
   n  \\
   1  \\
\end{array}} \right)+1=4n+1$. We have
\begin{equation}p^{n-k}\ge 4n+1,\end{equation} because all these
vectors must be elements of distinct cosets of the standard array
and we have $p^{n-k}$ cosets.

\end{proof}

To investigate the parameters of perfect codes, we must consider
the inequality (1) as
\begin{equation}p^{n-k}= 4n+1.\end{equation}
We now examine the values of $n$ and $k$ satisfying Eq. (2). Some
values of $n$ and $k$ satisfying Eq. (2) are $$\left( {n,k}
\right) = \left\{ {\left( {3,2} \right),\left( {4,3}
\right),\left( {6,4} \right),\left( {7,6} \right),\left( {9,8}
\right),\ldots ,\left( {31,28} \right),\left( {42,40}
\right),\ldots,\left( {549,546} \right),  \ldots } \right\}.$$
These values show that possible perfect codes over $G_\pi$
correcting all error patterns of Mannheim weight 1 are
$(3,2),(4,3),(6,4),(7,6),(9,8)...$ Note that one Mannheim error
correcting codes (OMEC) introduced by Huber in \cite{7} are
perfect. An OMEC
code have parameters $(n,n-1)$, where $n=(p-1/4)$.\\

We suppose that $p$ equals 5 in Eq. (2). Then, we have
\begin{equation}5^{n-k}= 4n+1.\end{equation} An integral solution of Eq. (3) is $(6,4)$. In the following, we
give an example of a $(6,4)$ perfect code correcting errors of
Mannheim weight 1. The $(6,4)$ code is not an OMEC code.

\begin{example} Consider the following parity check matrix for
$(6,4)$ perfect code over $G_{2+i}$: $$H = \left[
{\begin{array}{*{20}{c}}
   1 & 0 & 1 & 1 & i &  1  \\
   0 & 1 &1 & -1 & 1 & i \\
\end{array}} \right].$$
The code which is the null space of $H$ can correct all errors of
Mannheim weight 1 over $G_{2+i}$. In Table I, we give all the
error vectors of Mannheim weight 1 and their corresponding
syndromes over $G_{2+i}$ which can be seen to be distinct
altogether and detailed.
\end{example}

\begin{center}
{\scriptsize {Table I: Error patterns of Mannheim weight 1 and
their corresponding syndromes.}} {\small \centering}
\begin{tabular}{l  r}
  \hline
  Error pattern & Syndrome \\
  \hline
  $(100000)$ &  $(1,0)$ \\
  $(-100000)$ & $(-1,0)$ \\
  $(i00000)$ & $(i,0)$ \\
  $(-i00000)$ & $(-i,0)$ \\
  $(010000)$ &$ (0,1)$ \\
  $(0-10000)$ & $(0,-1)$ \\
  $(0i0000)$ & $(0,i)$ \\
  $(0-i0000)$ & $(0,-i)$ \\
  $(001000)$ & $(1,1)$ \\
  $(00-1000)$ & $(-1,-1)$ \\
  $(00i000)$ & $(i,i)$ \\
  $(00-i000)$ & $(-i,-i)$ \\
  $(000100)$ & $(1,-1)$ \\
  $(000-100)$ & $(-1,1)$ \\
  $(000i00)$ & $(i,-i)$ \\
  $(000-i00)$ & $(-i,i)$ \\
  $(000010) $& $(i,1)$ \\
  $(0000-10)$ & $(-i,-1)$ \\
  $(0000i0)$ &$ (-1,i)$ \\
  $(0000-i0)$ &$ (1,-i)$ \\
  $(000001) $& $(1,i) $\\
  $(00000-1) $& $(-1,-i)$ \\
  $(00000i) $&$ (i,-1)$ \\
  $(00000-i)$ &$ (-i,1)$ \\
  \hline
\end{tabular}

\end{center}

Therefore, $[6,4,3]$ code is a perfect code over $G_{2+i}$
correcting errors of Mannheim weight 1.

\subsection{Perfect codes correcting errors of Mannheim weight
2 or less }

In this section, we get a bound for an $(n, k)$ linear code which
corrects all error patterns of Mannheim weight 2 or less over
$G_{2+i}$ and $G_\pi$ ($\pi \pi ^*=p\ge 13$, a prime). Hence, we
obtain possible perfect codes. In this sequence, the first theorem
is as follows.

\begin{theorem}An $(n, k)$ linear code over $G_{2+i}$ corrects all errors of Mannheim weight 2 or less
provided that the bound \begin{equation}5^{n-k}\ge
8n^2-4n+1.\end{equation}
\end{theorem}

\begin{proof}We first enumerate error vectors of Mannheim weight 2 or
less. The number of error vectors of Mannheim weight 1 including
the vector of all zeros over $G_{2+i}$ is $4n+1$. There are only
one type error vectors of Mannheim weight 2 over $G_{2+i}$.

Those vectors that have two nonzero components and the nonzero
components could be one of the four
values $\pm 1$, $\pm i$.\\

The number of such vectors is $16\left( {\begin{array}{*{20}{c}}
   n  \\
   2  \\
\end{array}} \right)=8n^2-8n.$ \\

Thus, total number of error vectors of Mannheim weight 2 or
less over $G_{2+i} $ is equal to $8n^2-4n+1$. Also, the number of available cosets is equal to $5^{n-k}$. \\
Therefore, in order to correct all errors of Mannheim weight 2 or
less, the code must satisfy $5^{n-k}\ge 8n^2-4n+1$. Hence, the
proof is completed.
\end{proof}

To obtain the parameters of perfect codes, we must consider the
inequality (4) as
\begin{equation}5^{n-k}= 8n^2-4n+1.\end{equation} The integral solutions of
Eq. (5) for $n$ and $k$ are $(1,0)$, $(2,0)$. The solution
$(1,0)$, $(2,0)$ are not feasible since $k$ must greater than or
equal to $1$. So, we conclude that there doses not exists a
perfect code over $G_{2+i} $ correcting all errors of Mannheim
weight 2 or less.

\begin{theorem}An $(n, k)$ linear code over $G_\pi$ corrects all errors of Mannheim weight 2 or less
provided that the bound \begin{equation}p^{n-k}\ge
8n^2+1,\end{equation} where $p\equiv 1 \ (mod\ 4), p\ge13$.

\end{theorem}

\begin{proof}We first enumerate error vectors of Mannheim weight 1.\\

The number of error vectors of Mannheim weight 1 including the
vector of all zeros over $G_\pi$ is $4n+1$.
\\

There are two types error vectors of Mannheim weight 2 over $G_\pi$. \\

$(1)$ Those vectors that have two nonzero components and the
nonzero components could be one of the four values $\pm 1$, $\pm i$.\\

The number of such vectors is $  16\left( {\begin{array}{*{20}{c}}
   n  \\
   2  \\
\end{array}} \right)=8n^2-8n.$ \\

$(2)$ Those error vectors that have only one nonzero component and
the nonzero component could be one of the four values $\pm 2$,
$\pm 2i$.\\

The number of such vectors is $4n$.

Thus, total number of error vectors of Mannheim weight 2 or less
over $G_\pi$ is equal to $$8n^2+1.$$

Also, the number of available cosets is $p^{n-k}$. \\
Therefore, in order to correct all error vectors of Mannheim
weight 2 or less, the code must satisfy the bound $p^{n-k}\ge
8n^2+1$. Hence, the proof is completed.
\end{proof}

To obtain the parameters of perfect codes, we must consider the
inequality (6) as
\begin{equation}p^{n-k}= 8n^2+1.\end{equation} One can obtain that some integral solutions of
Eq. (7) for $n$, $k$ are $(3,2)$, $(6,4)$, $(12,11)$, $(15,14)$,
$(18,17)$, $(21,20)$, $(33,32)$, ...,$(204,202)$,...

These values show that possible perfect codes over $G_\pi$
correcting all errors of Mannheim weight 2 or less are $(3,2)$,
$(6,4)$, $(12,11)$, $(15,14)$, $(18,17)$, $(21,20)$, $(33,32)$,
...,$(204,202)$,...

\ Using a computer programme, for $n-k=1$, we show that there does
not exist any perfect code correcting errors of Mannheim weight 2
or less. However, the existence/nonexistence of perfect codes
correcting errors of Mannheim weight 2 or less over $G_\pi$
($n-k\ge 2$) is still unknown (except some special works
\cite{7,9}).

\section{Perfect codes over Lipschitz integers with respect to Lipschitz metric}

\subsection{Perfect codes correcting errors of Lipschitz weight
1}

We first obtain an upper bound on the number of parity check
digits for one Lipschitz error correcting codes over
$H(\mathcal{Z})_\pi$. Note that a Lipschitz error of weight 1
takes on one of the eight values $\pm 1,\ \pm e_1, \ \pm e_2,\ \pm
e_3$, at position $l$($0 \le l \le n - 1$ ).

\begin{theorem} An $(n, k)$ linear code over $H(\mathcal{Z})_\pi$ corrects all errors of Lipschitz weight
1 provided that $(p^2)^{n-k}\ge 8n+1$, where $p=\pi \pi ^*$ and
$p$ is a prime integer.

\end{theorem}

\begin{proof} We know that the cardinal number of $H(\mathcal{Z})_\pi$ is $p^2$ (see Thm. 1). Error vectors of Lipschitz weight one are
those vectors which have only one nonzero component and the
nonzero component could be one of the eight element  $\pm 1,\  \pm
e_1,\ \pm e_2, \pm e_3$.

The number of such vectors is equal to $8n$. Therefore, the number
of error vectors of Lipschitz weight 1 including the vector of all
zeros is equal to $  8\left( {\begin{array}{*{20}{c}}
   n  \\
   1  \\
\end{array}} \right)+1=8n+1$.

Since all these vectors must elements of distinct cosets of the
standard array and we have $(p^2)^{n-k}$ cosets in all, therefore,
we obtain \begin{equation}(p^2)^{n-k}\ge 8n+1.\end{equation}
Hence, the proof is completed.

\end{proof}

To obtain the parameters of perfect codes, we must consider the
inequality (8) as

\begin{equation}(p^2)^{n-k}= 8n+1.\end{equation}

A set of some integral solutions of Eq. (9) is
$$\left( {n,k} \right) = \left\{ {\left( {3,2} \right),\left(
{6,5} \right),\left( {10,8} \right),\ldots ,\left( {3570,3568}
\right), \ldots } \right\}.$$ These values show that the
parameters of possible perfect codes correcting all error patterns
of Lipschitz weight 1 and no others over $H(\mathcal{Z})_\pi$ are
$(3,2),(6,5),(10,8),(15,14),(21,20)...$.

We suppose that $p$ is equal to 5 in Eq. (9). Then, we get
\begin{equation}(5^2)^{n-k}= (1+8n).\end{equation}

The integral solutions of Eq. (10) for $n$ and $k$ are $$(3,2),\
(1953,1950),...$$ These values show the possibility of the
existence of $(3,2),\ (1953,1950),...$ perfect codes correcting
errors of Lipschitz weight 1 over $H(\mathcal{Z})_{2+e_1}$.

In the following, we give an example of a $(3,2)$ perfect code
correcting errors of Lipschitz weight 1 over
$H(\mathcal{Z})_{2+e_1}$.

\begin{example} Let $C$ be a code defined by the parity check matrix

$$H = \left[ {\begin{array}{*{20}{c}}
   1, & {1 + {e_3}}, & {1 + {e_2}}  \\
\end{array}} \right].$$
The code $C$ have the parameters $(3,2)$. It can correct all
errors of Lipschitz weight 1. In Table II, we give all the error
vectors of Lipschitz weight 1 and their corresponding syndromes
over $H(\mathcal{Z})_{2+e_1}$.

\begin{center}
{\scriptsize {Table II: Error patterns of Lipschitz weight 1 and
their corresponding syndromes.}} {\small \centering}
\begin{tabular}{l       r}
  \hline
  Error pattern & Syndrome \\
  \hline
  $(1,0,0)$ & $1$ \\
  $(e_1,0,0)$ & $e_1$ \\
  $(e_2,0,0)$ & $e_2$ \\
  $(e_3,0,0)$ & $e_3$ \\
  $(-1,0,0)$ & $-1$ \\
  $(-e_1,0,0)$ & $-e_1$ \\
  $(-e_2,0,0)$ & $-e_2$ \\
  $(-e_3,0,0)$ & $-e_3$ \\
  $(0,1,0)$ & $1+e_3$ \\
  $(0,e_1,0)$ & $e_1-e_2$ \\
  $(0,e_2,0)$ & $e_1+e_2$ \\
  $(0,e_3,0)$ & $1+e_3$ \\
  $(0,-1,0)$ & $-1-e_3$ \\
  $(0,-e_1,0)$ & $-e_1+e_2$ \\
  $(0,-e_2,0)$ & $-e_1-e_2$ \\
  $(0,-e_3,0)$ & $1-e_3$ \\
  $(0,0,1)$ & $1+e_2$ \\
  $(0,0,e_1)$ & $e_1+e_3$ \\
  $(0,0,e_2)$ & $-1+e_2$ \\
  $(0,0,e_3)$ & $-e_1+e_3$ \\
  $(0,0,-1)$ & $-1-e_2$ \\
  $(0,0,-e_1)$ & $-e_1-e_3$ \\
  $(0,0,-e_2)$ & $1-e_2$ \\
  $(0,0,-e_3)$ & $e_1-e_3$ \\
  \hline
\end{tabular}

\end{center}

\end{example}
The code with parameters $[3,2,3]$ over $H(\mathcal{Z})_{2+e_1}$
is a perfect code correcting all errors of Lipschitz weight 1.

\begin{remark}
We have investigated solutions of Eq. (9) for $p = 5$. We have
been able to obtain a perfect code for one of the solutions. One
can similarly solve the existence of perfect codes correcting
errors of Lipschitz weight 1 over
$H(\mathcal{Z})_{2+e_1+e_2+e_3}$, $H(\mathcal{Z})_{1+e_1+e_2}$,
$H(\mathcal{Z})_{3+e_1+e_2}$, ... by taking $\pi =2+e_1+e_2+e_3, \
1+e_1+e_2, \ 3+e_1+e_2,...$, respectively, in Eq. (9) and finding
the solutions for $n$ and $k$.
\end{remark}
Note that for $n-k=1$, there always exists a perfect code
corresponding to the parameters obtained by Eq. (9).

\subsection{Perfect codes correcting errors of Lipschitz weight
2 or less }

In this section, we obtain bound on the number of parity check
digits for an $(n, k)$ linear code correcting all error patterns
of Lipschitz weight 2 or less over $H(\mathcal{Z})_{1+e_1+e_2}$,
$H(\mathcal{Z})_{2+e_1}$, $H(\mathcal{Z})_{2+e_1+e_2+e_3}$,
$H(\mathcal{Z})_{3+e_1+e_2}$ and over $H(\mathcal{Z})_\pi$ ($\pi
\pi ^*=p\ge 13$, a prime) and then we investigate the existence of
corresponding perfect codes. In this sequence, the first theorem
is as follows.

\begin{theorem} An (n, k) linear code over $H(\mathcal{Z})_{1+e_1+e_2}$ corrects all errors
of Lipschitz weight 2 or less provided that the bound
$(p^2)^{n-k}\ge 32n^2-24n+1$.

\end{theorem}

\begin{proof} We first enumerate error vectors of Lipschitz weight 2 or less over
$H(\mathcal{Z})_{1+e_1+e_2}$. \\

The number of error vectors of Lipschitz weight 1 including the
vector of all zeros over
$H(\mathcal{Z})_{1+e_1+e_2}$ is $8n+1$. \\

 There is only one type error vectors of Lipschitz weight 2 over
 $H(\mathcal{Z})_{1+e_1+e_2}$.

Those vectors which have two nonzero components and the nonzero
components could be one of the eight values $\pm 1$, $\pm e_1$,
$\pm e_2$, $\pm e_3$.

The number of such vectors is equal to $64\left(
{\begin{array}{*{20}{c}}
   n  \\
   2  \\
\end{array}} \right) = 32{n^2} - 32n$.

Thus, total number of error vectors of Lipschitz weight 2 or less
over $H(\mathcal{Z})_{1+e_1+e_2}$ is equal to $32n^2-24n+1$. Also,
the number of available cosets is $(3^2)^{n-k}$. In order to
correct all error patterns of Lipschitz 2 or less over
$H(\mathcal{Z})_{1+e_1+e_2}$, the code must satisfy
\begin{equation}(3^2)^{n-k}\ge 32n^2-24n+1.\end{equation} Hence, the proof is completed.
\end{proof}

To obtain the parameters of perfect codes, we must consider the
inequality (11) as
\begin{equation}(3^2)^{n-k}= 32n^2-24n+1.\end{equation}

The integral solutions of Eq. (12) are $n = 1, k = 0$ and  $n = 2,
k = 0$. the solution $n = 1, k = 0$ and  $n = 2, k = 0$ are not
feasible as $n\ge 2$ and $k>0$. So, we conclude that there does
not exists a perfect code over $H(\mathcal{Z})_{1+e_1+e_2}$
correcting all error patterns of Lipschitz weight 2 or less.

Now, we obtain the bound on the number of parity check digits for
an $(n, k)$ linear code over $H(\mathcal{Z})_{2+e_1}$,
$H(\mathcal{Z})_{2+e_1+e_2+e_3}$ and $H(\mathcal{Z})_{3+e_1+e_2}$
correcting errors of Lipschitz weight 2 or less.

\begin{theorem} An $(n, k)$ linear code over $H(\mathcal{Z})_{2+e_1}$ corrects all errors of Lipschitz weight 2 or
less provided that the bound $(p^2)^{n-k}\ge 32n^2-8n+1$.

\end{theorem}

\begin{proof} We first enumerate error vectors of Lipschitz weight 2 or less over
$H(\mathcal{Z})_{2+e_1}$. \\

The number of error vectors of Lipschitz weight 1 including the
vector of all zeros over $H(\mathcal{Z})_{2+e_1}$ is $8n+1$. \\

There are two types error vectors of Lipschitz weight 2 over
$H(\mathcal{Z})_{2+e_1}$.

$(1)$ Those vectors which are also error vectors of Lipschitz
weight 2 over $H(\mathcal{Z})_{1+e_1+e_2}$.

The number of such vectors is equal to $64\left(
{\begin{array}{*{20}{c}}
   n  \\
   2  \\
\end{array}} \right) = 32{n^2} - 32n$.

$(2)$ Those error vectors which have only one nonzero component
and the nonzero component could be one of the sixteen values $\pm
(1+e_2), \ \pm (1+e_3), \pm (e_1+e_2), \ \pm (e_1+e_3)$, $\pm
(1-e_2), \ \pm (1-e_3), \  \pm (e_1-e_2), \ \pm (e_1-e_3)$.\\

The number of such vectors is equal to $16n$.

Thus, total number of error vectors of Lipschitz 2 or less over
$H(\mathcal{Z})_{2+e_1}$ is equal to $32n^2-8n+1$. Also, the
number of available cosets is $(5^2)^{n-k}$. In order to correct
all error patterns of Lipschitz weight 2 or less over
$H(\mathcal{Z})_{2+e_1}$, the code must satisfy
\begin{equation}(5^2)^{n-k}\ge 32n^2-8n+1.\end{equation} Hence, the proof is completed.
\end{proof}

To obtain the parameters of perfect codes, we must consider the
inequality (13) as \begin{equation}(5^2)^{n-k}=
32n^2-8n+1.\end{equation}

The only integral solution of Eq. (14) is $n = 1, k = 0$. The
solution $n = 1, k = 0$ is not feasible. So, we conclude that
there does not exists a perfect code over $H(\mathcal{Z})_{2+e_1}$
correcting all error patterns of Lipschitz weight 2 or less.

\begin{theorem} An $(n, k)$ linear code over $H(\mathcal{Z})_{2+e_1+e_2+e_3}$ and $H(\mathcal{Z})_{3+e_1+e_2}$
corrects all errors of Lipschitz weight 2 or less provided that
the bound $(p^2)^{n-k}\ge 32n^2+1$.

\end{theorem}

\begin{proof} We first enumerate error vectors of Lipschitz weight 2 or less over
$H(\mathcal{Z})_{2+e_1+e_2+e_3}$ and $H(\mathcal{Z})_{3+e_1+e_2}$. \\

The number of error vectors of Lipschitz weight 1 including the
vector of all zeros over $H(\mathcal{Z})_{2+e_1+e_2+e_3}$ and $H(\mathcal{Z})_{3+e_1+e_2}$ is equal to $8n+1$. \\

There are two types error vectors of Lipschitz weight 2 over
 $H(\mathcal{Z})_{2+e_1+e_2+e_3}$ and
$H(\mathcal{Z})_{3+e_1+e_2}$.

$(1)$ Those vectors which are also error vectors of Lipschitz
weight 2 over $H(\mathcal{Z})_{1+e_1+e_2}$.

The number of such vectors is $64\left( {\begin{array}{*{20}{c}}
   n  \\
   2  \\
\end{array}} \right) = 32{n^2} - 32n$.

$(2)$ Those error vectors which have only one nonzero component
and the nonzero component could be one of the twenty four  values
$\pm (1+e_1), \ \pm (1+e_2), \ \pm (1+e_3),\ \pm (e_1+e_2),\ \pm
(e_1+e_3),\ \pm (e_2+e_3)$, $\pm (1-e_1), \
\pm (1-e_2), \ \pm (1-e_3),\ \pm (e_1-e_2),\ \pm (e_1-e_3),\ \pm (e_2-e_3)$.\\

The number of such vectors is $24n$.

Thus, total number of error vectors of Lipschitz 2 or less over
$H(\mathcal{Z})_{2+e_1+e_2+e_3}$ and $H(\mathcal{Z})_{3+e_1+e_2}$
is equal to $32n^2+1$. Also, the number of available cosets is
equal to $(7^2)^{n-k}$ and $(11^2)^{n-k}$, respectively. In order
to correct all error patterns of the Lipschitz weight 2 or less
over $H(\mathcal{Z})_{2+e_1+e_2+e_3}$ and
$H(\mathcal{Z})_{3+e_1+e_2}$, the code must satisfy
\begin{equation}(7^2)^{n-k}\ge  32n^2+1, \ (11^2)^{n-k}\ge  32n^2+1,\end{equation} respectively. Hence, the proof is completed.
\end{proof}

To obtain the parameters of perfect codes, we must consider the
inequality (15) as \begin{equation}(7^2)^{n-k}= 32n^2+1,\
(11^2)^{n-k}= 32n^2+1.\end{equation}

There is no integral solution of Eq. (16) and thus no perfect code
exists in this case.

\begin{theorem} An $(n, k)$ linear code over $H(\mathcal{Z})_\pi$ ($\pi \pi ^*=p\ge 13$ a prime) corrects all errors of Lipschitz weight 2 or less
provided that the bound $(p^2)^{n-k}\ge 32n^2+8n+1$.
\end{theorem}

\begin{proof} The number of error vectors of the Lipschitz weight 1 including the vector of all zeros over
$H(\mathcal{Z})_\pi$ ($\pi \pi ^*=p\ge 13$) is $8n+1$. \\

There are two types error vectors of Lipschitz weight two over
$H(\mathcal{Z})_\pi$ ($\pi \pi ^*=p\ge 13$)  \\

$(1)$ Those vectors which are also error vectors of Lipschitz
weight 2 over $H(\mathcal{Z})_{2+e_1}$,
$H(\mathcal{Z})_{2+e_1+e_2+e_3}$ and $H(\mathcal{Z})_{3+e_1+e_2}$. The number of such vectors is $32n^2-32n.$ \\

$(2)$ Those error vectors which have only one nonzero component
and the nonzero component could be one of the Thirty-two values
$\pm 2, \ \pm 2e_1, \pm 2e_2,\ \pm 2e_3$, $\pm (1+e_1), \ \pm
(1+e_2), \ \pm (1+e_3),\ \pm (e_1+e_2),\ \pm (e_1+e_3),\ \pm
(e_2+e_3)$, $\pm (1-e_1), \
\pm (1-e_2), \ \pm (1-e_3),\ \pm (e_1-e_2),\ \pm (e_1-e_3),\ \pm (e_2-e_3)$. \\

The number of such vectors is $32n$.\\

Thus, total number of error vectors of Lipschitz 2 or less over
$H(\mathcal{Z})_\pi$ ($\pi\pi ^*=p\ge 13$) is equal to
$32n^2+8n+1$. Also, the number of available cosets is equal to
$(p^2)^{n-k}$. In order to correct all error patterns of Lipschitz
2 or less over $H(\mathcal{Z})_\pi$ , the code must satisfy
\begin{equation}(p^2)^{n-k} \ge  32n^2+8n+1.\end{equation} Hence, the proof is completed.

\end{proof}

To obtain the parameters of perfect codes, we must consider the
inequality (17) as
\begin{equation}(p^2)^{n-k}= 32n^2+8n+1.\end{equation}

Take $p =29$ in Eq. (18), we get \begin{equation}(29^2)^{n-k}=
32n^2+8n+1.\end{equation} The only integral solution of Eq. (19)
for $n$ and $k$ is $n=5, \ k=4$.\\

Take $p =33461$ in Eq. (18), we get
\begin{equation}(33461^2)^{n-k}= 32n^2+8n+1.\end{equation} The
only integral solution of Eq. (20)
for $n$ and $k$ is $n=5915, \ k=5914$.\\

There is no other integral solution of Eq. (18) other than the
above mentioned solutions.\\

These values show the possibility of the existence of $(5, 4),\
(5915, 5914)$ ($n\le 10000$) perfect codes over
$H(\mathcal{Z})_\pi$ ($\pi \pi ^*=p\ge 13$) correcting all error
patterns of Lipschitz weight 2 or less and no others.

\section{Perfect codes over Hurwitz integers with respect to Lipschitz metric}
\subsection{Perfect codes correcting errors of Lipschitz weight
1 over Hurwitz integers}

We first obtain an upper bound on the number of parity check
digits for one Lipschitz error correcting codes over
$\mathcal{H}_\pi$. Note that a Lipschitz error of weight 1 takes
on one of the eight values $\pm 1,\ \pm e_1, \ \pm e_2,\ \pm e_3$,
at position $l$($0 \le l \le n - 1$ ).

\begin{theorem} An $(n, k)$ linear code over $\mathcal{H}_\pi$ corrects all errors of Lipschitz weight
1 provided that $(2p^2-1)^{n-k}\ge 8n+1$, where $p=\pi \pi ^*$ and
$p$ is a prime integer.

\end{theorem}

\begin{proof} We know that the cardinal number of $\mathcal{H}_\pi$ is $2p^2-1$ (see Thm. 2). Error vectors of Lipschitz weight one are
those vectors which have only one nonzero component and the
nonzero component could be one of the eight element  $\pm 1,\  \pm
e_1,\ \pm e_2, \pm e_3$.

The number of such vectors is equal to $8n$. Therefore, the number
of error vectors of Lipschitz weight 1 including the vector of all
zeros is equal to $  8\left( {\begin{array}{*{20}{c}}
   n  \\
   1  \\
\end{array}} \right)+1=8n+1$.

Since all these vectors must elements of distinct cosets of the
standard array and we have $(2p^2-1)^{n-k}$ cosets in all,
therefore, we obtain \begin{equation}(2p^2-1)^{n-k}\ge
8n+1.\end{equation} Hence, the proof is completed.

\end{proof}

To obtain the parameters of perfect codes, we must consider the
inequality (21) as

\begin{equation}(2p^2-1)^{n-k}= 8n+1.\end{equation}

We suppose that $p$ is equal to 3 in Eq. (22). Then, we get
\begin{equation}(17)^{n-k}= (1+8n).\end{equation}

The integral solutions of Eq. (23) for $n$ and $k$ are $$(2,1),\
(36,34),(614,611),(10440,10436),(177482,177477),...$$

We suppose that $p$ is equal to 5 in Eq. (22). Then, we get
\begin{equation}(49)^{n-k}= (1+8n).\end{equation}

The integral solutions of Eq. (24) for $n$ and $k$ are $$(6,5),\
(300,298),(14706,14703),(720600,720596),...$$

We suppose that $p$ is equal to 7 in Eq. (22). Then, we get
\begin{equation}(97)^{n-k}= (1+8n).\end{equation}

The integral solutions of Eq. (25) for $n$ and $k$ are $$(12,11),\
(1176,1174),(114084,114081),...$$

There always exist a perfect code which its parameters
corresponding to above parameters. These perfect codes are not
known before.

In the following, we give an example of a $(2,1)$ perfect code
correcting errors of Lipschitz weight 1 over
$\mathcal{H}_{1+e_1+e_2}$.

\begin{example}Consider the following parity check matrix $H$ for (2, 1) perfect code
over $\mathcal{H}_{1+e_1+e_2}$:

$$H = \left[ {\begin{array}{*{20}{c}}
   1, & {\frac{1}{2} + \frac{e_1}{2}+ \frac{e_2}{2}+ \frac{e_3}{2}}  \\
\end{array}} \right].$$
The code which is the null space of $H$ can correct all errors of
Lipschitz weight 1 over $\mathcal{H}_{1+e_1+e_2}$ and no others.
In Table III, we list all the error vectors of Lipschitz weight 1
and their corresponding syndromes over $\mathcal{H}_{1+e_1+e_2}$
which can be seen to be distinct altogether and exhaustive.\\

\begin{center}
{\scriptsize {Table III: Error patterns of Lipschitz weight 1 and
their corresponding syndromes.}} {\small \centering}
\begin{tabular}{l       r}
  \hline
  Error pattern & Syndrome \\
  \hline
  $(1,0)$ & $1$ \\
  $(e_1,0)$ & $e_1$ \\
  $(e_2,0)$ & $e_2$ \\
  $(e_3,0)$ & $e_3$ \\
  $(-1,0)$ & $-1$ \\
  $(-e_1,0)$ & $-e_1$ \\
  $(-e_2,0)$ & $-e_2$ \\
  $(-e_3,0)$ & $-e_3$ \\
  $(0,1)$ & ${\frac{1}{2} + \frac{e_1}{2}+ \frac{e_2}{2}+ \frac{e_3}{2}}$ \\
  $(0,e_1)$ & $-{\frac{1}{2} + \frac{e_1}{2}- \frac{e_2}{2}+ \frac{e_3}{2}}$ \\
  $(0,e_2)$ & $-{\frac{1}{2} + \frac{e_1}{2}+ \frac{e_2}{2}- \frac{e_3}{2}}$ \\
  $(0,e_3)$ & $-{\frac{1}{2} - \frac{e_1}{2}+ \frac{e_2}{2}+ \frac{e_3}{2}}$ \\
  $(0,-1)$ & $-{\frac{1}{2} - \frac{e_1}{2} - \frac{e_2}{2} - \frac{e_3}{2}}$ \\
  $(0,-e_1)$ & ${\frac{1}{2} - \frac{e_1}{2}+ \frac{e_2}{2}- \frac{e_3}{2}}$ \\
  $(0,-e_2)$ & ${\frac{1}{2} - \frac{e_1}{2}- \frac{e_2}{2}+ \frac{e_3}{2}}$ \\
  $(0,-e_3)$ & ${\frac{1}{2} + \frac{e_1}{2}- \frac{e_2}{2}- \frac{e_3}{2}}$ \\

  \hline
\end{tabular}

\end{center}
Therefore, $(2, 1)$ code is a perfect code correcting errors of
Lipschitz weight 1 over $\mathcal{H}_{1+e_1+e_2}$ .

\end{example}
To the best of our knowledge, above perfect code is not known
before.

\subsection{Perfect codes correcting errors of Lipschitz weight
2 or less over Hurwitz integers}

In this section, we obtain bound on the number of parity check
digits for an $(n, k)$ linear code correcting all error patterns
of Lipschitz weight 2 or less over $\mathcal{H}_{1+e_1+e_2}$,
$\mathcal{H}_{2+e_1}$,  $\mathcal{H}_{2+e_1+e_2+e_3}$,
$\mathcal{H}_{3+e_1+e_2}$ and $\mathcal{H}_\pi$ ($\pi \pi ^*=p\ge
7$, a prime). In this sequence, the first theorem is as follows.

\begin{theorem} An (n, k) linear code over $\mathcal{H}_{1+e_1+e_2}$ corrects all errors
of Lipschitz weight 2 or less provided that the bound $17^{n-k}\ge
32n^2-16n+1$.

\end{theorem}

\begin{proof} We first enumerate error vectors of Lipschitz weight 2 or less over
$\mathcal{H}_{1+e_1+e_2}$. \\

The number of error vectors ofLipschitz weight 1 including the
vectors of all zeros over $\mathcal{H}_{1+e_1+e_2}$ is equal to $8n+1$. \\

There are two types error vectors of Lipschitz weight 2 over
$\mathcal{H}_{1+e_1+e_2}$. \\

$(1)$ Those vectors which have two nonzero components and the
nonzero components could be one of the eight values $\pm 1$, $\pm
e_1$, $\pm e_2$, $\pm e_3$.

The number of such vectors is $64\left( {\begin{array}{*{20}{c}}
   n  \\
   2  \\
\end{array}} \right) = 32{n^2} - 32n$.

$(2)$ Those error vectors which have only one nonzero component
and the nonzero component could be one of the eight values $ \frac
{1}{2} + \frac {e_1}{2} + \frac {e_2}{2} + \frac {e_3}{2}$, $
-\frac {1}{2} + \frac {e_1}{2} + \frac {e_2}{2} + \frac {e_3}{2}$,
$  \frac {1}{2} - \frac {e_1}{2} + \frac {e_2}{2} + \frac
{e_3}{2}$, $  \frac {1}{2} + \frac {e_1}{2} - \frac {e_2}{2} +
\frac {e_3}{2}$, $  \frac {1}{2} + \frac {e_1}{2} + \frac {e_2}{2}
- \frac {e_3}{2}$, $  -\frac {1}{2} - \frac {e_1}{2} + \frac
{e_2}{2} + \frac {e_3}{2}$,
$  -\frac {1}{2} + \frac {e_1}{2} - \frac {e_2}{2} + \frac {e_3}{2}$, $  -\frac {1}{2} + \frac {e_1}{2} + \frac {e_2}{2} - \frac {e_3}{2}$. \\

The number of such vectors is $8n$.

Thus, total number of error vectors of Lipschitz 2 or less over
$\mathcal{H}_{1+e_1+e_2}$ is $32n^2-16n+1$. Also, the number of
available cosets is equal to $17^{n-k}$. In order to correct all
error patterns of Lipschitz weight 2 or less over
$\mathcal{H}_{1+e_1+e_2}$, the code must satisfy
\begin{equation}17^{n-k}\ge 32n^2-16n+1.\end{equation} Hence, the proof is completed.
\end{proof}

To obtain the parameters of perfect codes, we must consider the
inequality (26) as \begin{equation}17^{n-k}=
32n^2-16n+1.\end{equation} The only integral solution of Eq. (27)
is $n = 1; k = 0$. The solution $n = 1; k = 0$ is not feasible.
There is no other integral solution of Eq. (27) other than the
above mentioned solutions. So, we conclude that there is not exist
a perfect code over $\mathcal{H}_{1+e_1+e_2}$.

\begin{theorem} An (n, k) linear code over $\mathcal{H}_{2+e_1}$ corrects all errors
of Lipschitz weight 2 or less provided that the bound $49^{n-k}\ge
32n^2+8n+1$.

\end{theorem}

\begin{proof} We first enumerate error vectors of Lipschitz weight 2 or less over
$\mathcal{H}_{2+e_1}$. \\

The number of error vectors of Lipschitz weight 1 including the
vectors of all zeros over $\mathcal{H}_{2+e_1}$ is equal to $8n+1$. \\

There are two types error vectors of Lipschitz weight 2 over
$\mathcal{H}_{2+e_1}$. \\

$(1)$ Those vectors which have two nonzero components and the
nonzero components could be one of the eight values $\pm 1$, $\pm
e_1$, $\pm e_2$, $\pm e_3$.

The number of such vectors is $64\left( {\begin{array}{*{20}{c}}
   n  \\
   2  \\
\end{array}} \right) = 32{n^2} - 32n$.

$(2)$ Those error vectors which have only one nonzero component
and the nonzero component could be one of the thirty two values
$\pm \frac {1}{2} \pm \frac {e_1}{2} \pm \frac {e_2}{2} \pm \frac {e_3}{2}$, $\pm (1+e_2)$, $\pm (1-e_2)$,
$\pm (1+e_3)$, $\pm (1-e_3)$, $\pm (e_1+e_2)$, $\pm (e_1-e_2)$, $\pm (e_1+e_3)$, $\pm (e_1-e_3)$. \\

The number of such vectors is $32n$.

Thus, total number of error vectors of Lipschitz weight 2 or less
over $\mathcal{H}_{2+e_1}$ is $32n^2+8n+1$. Also, the number of
available cosets is equal to $49^{n-k}$. In order to correct all
error patterns of Lipschitz weight 2 or less over
$\mathcal{H}_{2+e_1}$, the code must satisfy
\begin{equation}49^{n-k}\ge 32n^2+8n+1.\end{equation} Hence, the proof is completed.
\end{proof}

To obtain the parameters of perfect codes, we must consider the
inequality (28) as
\begin{equation}49^{n-k}= 32n^2+8n+1.\end{equation}

There is no integral solution of Eq. (29). So, there does not
exists a perfect code over $\mathcal{H}_{2+e_1}$ correcting all
error patterns of Lipschitz weight 2 or less.

\begin{theorem} An (n, k) linear code over $\mathcal{H}_{2+e_1+e_2+e_3}$ and $\mathcal{H}_{3+e_1+e_2}$ corrects all errors
of Lipschitz weight 2 or less provided that the bound $97^{n-k}\ge
32n^2+16n+1$ and $241^{n-k}\ge 32n^2+16n+1$, respectively.

\end{theorem}

\begin{proof} We first enumerate error vectors of Lipschitz weight 2 or less over
$\mathcal{H}_{2+e_1+e_2+e_3}$ and $\mathcal{H}_{3+e_1+e_2}$. \\

The number of error vectors of Lipschitz weight 1 including the
vectors of all zeros over $\mathcal{H}_{2+e_1+e_2+e_3}$ and $\mathcal{H}_{3+e_1+e_2}$ is equal to $8n+1$. \\

There are two types error vectors of Lipschitz weight 2 over
$\mathcal{H}_{2+e_1+e_2+e_3}$ and $\mathcal{H}_{3+e_1+e_2}$. \\

$(1)$ Those vectors which have two nonzero components and the
nonzero components could be one of the eight values $\pm 1$, $\pm
e_1$, $\pm e_2$, $\pm e_3$.

The number of such vectors is $64\left( {\begin{array}{*{20}{c}}
   n  \\
   2  \\
\end{array}} \right) = 32{n^2} - 32n$.

$(2)$ Those error vectors which have only one nonzero component
and the nonzero component could be one of the forty values $\pm
\frac {1}{2} \pm \frac {e_1}{2} \pm \frac {e_2}{2} \pm \frac
{e_3}{2}$, $\pm (1+e_1)$, $\pm (1-e_1)$, $\pm (1+e_2)$, $\pm
(1-e_2)$, $\pm (1+e_3)$, $\pm (1-e_3)$, $\pm (e_1+e_2)$, $\pm (e_1-e_2)$, $\pm (e_1+e_3)$, $\pm (e_1-e_3)$, $\pm (e_2+e_3)$, $\pm (e_2-e_3)$. \\

The number of such vectors is $40n$.

Thus, total number of error vectors of Lipschitz 2 or less over
$\mathcal{H}_{2+e_1+e_2+e_3}$ and $\mathcal{H}_{3+e_1+e_2}$ is
$32n^2+8n+1$. Also, the number of available cosets are equal to
$97^{n-k}$ and $241^{n-k}$, respectively. In order to correct all
error patterns of Lipschitz weight 2 or less over
$\mathcal{H}_{2+e_1+e_2+e_3}$ and $\mathcal{H}_{3+e_1+e_2}$, the
code must satisfy
\begin{equation}97^{n-k}\ge 32n^2+16n+1\end{equation} and \begin{equation}241^{n-k}\ge 32n^2+16n+1,\end{equation} respectively. Hence, the proof is completed.
\end{proof}

To obtain the parameters of perfect codes, we must consider the
inequality (30) and (31) as
\begin{equation}97^{n-k}= 32n^2+16n+1,\end{equation} \begin{equation}241^{n-k}= 32n^2+16n+1.\end{equation}

There is no integral solution of Eq. (32) and (33). So, there does
not exists a perfect code over $\mathcal{H}_{2+e_1+e_2+e_3}$ and
$\mathcal{H}_{3+e_1+e_2}$ correcting all error patterns of
Lipschitz weight 2 or less.

Now, we obtain the bound on the number of parity check digits for
an $(n, k)$ linear code over $\mathcal{H}_\pi$ ($\pi \pi ^*=p\ge
13$ a prime) correcting errors of Lipschitz weight 2 or less.

\begin{theorem} An $(n, k)$ linear code over $\mathcal{H}_\pi$ ($\pi \pi ^*=p\ge 13$ a prime)
corrects all errors of the Lipschitz weight 2 or less provided
that the bound $(2p^2-1)^{n-k}\ge 32n^2+24n+1$.
\end{theorem}

\begin{proof} The number of error vectors of Lipschitz weight one
including the vector of all zeros over $\mathcal{H}_\pi$ ($\pi \pi ^*=p\ge 13$ is equal to $8n+1$. \\

There are two types error vectors of Lipschitz weight 2 over
$\mathcal{H}_\pi$ ($\pi \pi ^*=p\ge 13$. \\

$(1)$ Those vectors which have two nonzero components and the
nonzero components could be one of the eight values $\pm 1$, $\pm
e_1$, $\pm e_2$, $\pm e_3$. \\

The number of such vectors is $32n^2-32n$.

$(2)$ Those error vectors which have only one nonzero component
and the nonzero component could be one of the forty eight  values
$\pm 2, \ \pm 2e_1, \pm 2e_2,\ \pm 2e_3$, $\pm \frac {1}{2} \pm
\frac {e_1}{2} \pm \frac {e_2}{2} \pm \frac {e_3}{2}$, $\pm
(1+e_1)$, $\pm (1+e_2)$, $\pm (1+e_3)$, $\pm (e_1+e_2)$, $\pm
(e_1+e_3)$, $\pm (e_2+e_3)$, $\pm (1-e_1)$, $\pm (1-e_2)$, $\pm
(1-e_3)$, $\pm (e_1-e_2)$, $\pm (e_1-e_3)$, $\pm (e_2-e_3)$. \\

The number of such vectors is $48n$.\\

Thus, total number of error vectors of Lipschitz weight 2 or less
over $\mathcal{H}_\pi$ ($\pi \pi ^*=p\ge 13$) is ($32n^2+24n+1$).
Also, the number of available cosets is equal to $(2p^2-1)^{n-k}$.
In order to correct all error patterns of Lipschitz weight 2 or
less over $\mathcal{H}_\pi$, the code must satisfy
\begin{equation}(2p^2-1)^{n-k} \ge  32n^2+24n+1.\end{equation} Hence, the proof is completed.

\end{proof}

To obtain the parameters of perfect codes, we must consider the
inequality (34) as
\begin{equation}(2p^2-1)^{n-k}= 32n^2+24n+1.\end{equation}

There is no integral solution of Eq. (35) for $n\le 1000000$.

\section{Perfect codes over Hurwitz integers with respect to Hurwitz metric}
\subsection{Perfect codes correcting errors of Hurwitz weight
1 over Hurwitz integers}

We first obtain an upper bound on the number of parity check
digits for one Hurwitz error correcting codes over
$\mathcal{H}_\pi$. Note that a Hurwitz error of weight 1 takes on
one of the ten values $\pm 1,\ \pm e_1, \ \pm e_2,\ \pm e_3$, $\pm
w=\pm(\frac{1}{2}+\frac{1}{2}e_1+\frac{1}{2}e_2+\frac{1}{2}e_3)$
at position $l$($0 \le l \le n - 1$ ).

\begin{theorem} An $(n, k)$ linear code over $\mathcal{H}_\pi$ corrects all errors of Hurwitz weight
1 provided that $(2p^2-1)^{n-k}\ge 10n+1$, where $p=\pi \pi ^*$
and $p$ is a prime integer.

\end{theorem}

\begin{proof} Error vectors of Hurwitz weight one are
those vectors which have only one nonzero component and the
nonzero component could be one of the ten elements  $\pm 1,\  \pm
e_1,\ \pm e_2, \pm e_3, \pm
w=\pm(\frac{1}{2}+\frac{1}{2}e_1+\frac{1}{2}e_2+\frac{1}{2}e_3)$.

The number of such vectors is equal to $10n$. Therefore, the
number of error vectors of Hurwitz weight 1 including the vector
of all zeros is equal to $  10\left( {\begin{array}{*{20}{c}}
   n  \\
   1  \\
\end{array}} \right)+1=10n+1$.

Since all these vectors must elements of distinct cosets of the
standard array and we have $(2p^2-1)^{n-k}$ cosets in all,
therefore, we obtain \begin{equation}(2p^2-1)^{n-k}\ge
8n+1.\end{equation} Hence, the proof is completed.

\end{proof}

To obtain the parameters of perfect codes, we must consider the
inequality (36) as

\begin{equation}(2p^2-1)^{n-k}= 10n+1.\end{equation}

We suppose that $p$ is equal to 3 in Eq. (37). Then, we get
\begin{equation}17^{n-k}= 10n+1.\end{equation}

The integral solutions of Eq. (38) for $n$ and $k$ are
$$(83520,83516),\ (6975757440,6975757432),...$$ These values show the possibility
of the existence of $(83520,83516),\ (6975757440,6975757432),...$
perfect codes correcting errors of Hurwitz weight 1 over
$\mathcal{H})_{1+e_1+e_2}$.

We suppose that $p$ is equal to 5 in Eq. (37). Then, we get
\begin{equation}49^{n-k}= 10n+1.\end{equation}

The integral solutions of Eq. (39) for $n$ and $k$ are
$$(2400,2398),\ (5764800,5764796),...$$

There are always integral solutions of Eq. (37) for primes $p \ge
7$.

These values show the possibility of the existence of
$(2400,2398),\ (5764800,5764796),...$ perfect codes correcting
errors of Hurwitz weight 1 over $\mathcal{H}_{2+e_1}$.

\subsection{Perfect codes correcting errors of Hurwitz weight
2 or less over Hurwitz integers}

In this section, we obtain bound on the number of parity check
digits for an $(n, k)$ linear code correcting all error patterns
of Hurwitz weight 2 or less over $\mathcal{H}_{1+e_1+e_2}$ and
$\mathcal{H}_\pi$ ($\pi \pi ^*=p\ge 17$, a prime). In this
sequence, the first theorem is as follows.

\begin{theorem} An (n, k) linear code over $\mathcal{H}_{1+e_1+e_2}$ corrects all errors
of Hurwitz weight 2 or less provided that the bound $17^{n-k}\ge
50n^2-34n+1$.

\end{theorem}

\begin{proof} We first enumerate error vectors of Hurwitz weight 2 or less over
$\mathcal{H}_{1+e_1+e_2}$. \\

The number of error vectors of Hurwitz weight 1 including the
vectors of all zeros over $\mathcal{H}_{1+e_1+e_2}$ is equal to $10n+1$. \\

There are two types error vectors of Hurwitz weight 2 over
$\mathcal{H}_{1+e_1+e_2}$. \\

$(1)$ Those vectors which have two nonzero components and the
nonzero components could be one of the eight values $\pm 1$, $\pm
e_1$, $\pm e_2$, $\pm e_3$.

The number of such vectors is $100\left( {\begin{array}{*{20}{c}}
   n  \\
   2  \\
\end{array}} \right) = 50{n^2} - 50n$.

$(2)$ Those error vectors which have only one nonzero component
and the nonzero component could be one of the six values $1-w,-1+w,i-w,-i+w,j-w,-j+w$. \\

The number of such vectors is $6n$.

Thus, total number of error vectors of Hurwitz weight 2 or less
over $\mathcal{H}_{1+e_1+e_2}$ is $50n^2-34n+1$. Also, the number
of available cosets is equal to $17^{n-k}$. In order to correct
all error patterns of Hurwitz weight 2 or less over
$\mathcal{H}_{1+e_1+e_2}$, the code must satisfy
\begin{equation}17^{n-k}\ge 50n^2-34n+1.\end{equation} Hence, the proof is completed.
\end{proof}

To obtain the parameters of perfect codes, we must consider the
inequality (40) as \begin{equation}17^{n-k}=
50n^2-34n+1.\end{equation} The only integral solution of Eq. (41)
is $n = 1; k = 0$. The solution $n = 1; k = 0$ is not feasible.
There is no other integral solution of Eq. (41) other than the
above mentioned solutions. So, we conclude that there is not exist
a perfect code over $\mathcal{H}_{1+e_1+e_2}$.

Now, we obtain the bound on the number of parity check digits for
an $(n, k)$ linear code over $\mathcal{H}_\pi$ ($\pi \pi ^*=p\ge
17$ a prime) correcting errors of Hurwitz weight 2 or less.

\begin{theorem} An $(n, k)$ linear code over $\mathcal{H}_\pi$ ($\pi \pi ^*=p\ge 13$ a prime)
corrects all errors of the Hurwitz weight 2 or less provided that
the bound $(2p^2-1)^{n-k}\ge 50n^2+10n+1$.
\end{theorem}

\begin{proof} The number of error vectors of Hurwitz weight one
including the vector of all zeros over $\mathcal{H}_\pi$ ($\pi \pi ^*=p\ge 17$ is equal to $10n+1$. \\

There are two types error vectors of Hurwitz weight 2 over
$\mathcal{H}_\pi$ ($\pi \pi ^*=p\ge 17$. \\

$(1)$ Those vectors which have two nonzero components and the
nonzero components could be one of the ten values $\pm 1$, $\pm
e_1$, $\pm e_2$, $\pm e_3$, $\pm w$. \\

The number of such vectors is $50n^2-50n$.

$(2)$ Those error vectors which have only one nonzero component
and the nonzero component could be one of the fifty  values $\pm
2, \ \pm 2e_1, \pm 2e_2,\ \pm 2e_3 \ \pm1\pm w,\ \pm i \pm w, \pm
j \pm w, \pm k \pm w$,
$\pm1\pm i,\ \pm 1 \pm j, \pm 1 \pm k, \pm i \pm j, \pm i \pm k,\ \pm j \pm k, \pm 2 w $ \\

The number of such vectors is $50n$.\\

Thus, total number of error vectors of Hurwitz weight 2 or less
over $\mathcal{H}_\pi$ ($\pi \pi ^*=p\ge 17$) is ($50n^2+10n+1$).
Also, the number of available cosets is equal to $(2p^2-1)^{n-k}$.
In order to correct all error patterns of Hurwitz weight 2 or less
over $\mathcal{H}_\pi$, the code must satisfy
\begin{equation}(2p^2-1)^{n-k} \ge  50n^2+10n+1.\end{equation} Hence, the proof is completed.

\end{proof}

To obtain the parameters of perfect codes, we must consider the
inequality (42) as
\begin{equation}(2p^2-1)^{n-k}= 50n^2+10n+1.\end{equation}

There is no integral solution of Eq. (35) for $n\le 1000000$,
$n-k\le 23$.\

 \

If we restrict $\mathcal{H}$ to $\mathcal{R}=\left\{ {a + bw:a,b
\in \mathcal{Z},w = \frac{1}{2}(1 + e_1+e_2+e_3)} \right\}$ then,
we obtain perfect codes corresponding to codes given in
\cite{Neto}.

\section{Conclusion} In this paper, we have investigated the existence/nonexistence of perfect codes correcting errors
of Mannheim, Lipschitz, and Hurwitz weight 1, errors of Mannheim
weight 2 or less, Lipschitz 2 or less, and Hurwitz weight 2 or
less over $G_\pi$, $H(\mathcal{Z})_\pi$ and $\mathcal{H}_\pi$. We
have been able to obtain perfect codes correcting errors of
Mannheim, Lipschitz, and Hurwitz weight 1. To the best of our
knowledge, some of these codes are not known before.

\end{document}